\documentclass[11pt]{article}
\usepackage{latexsym,amsmath,url,epsfig}
\usepackage{tikz}
\usepackage{textcomp}
\usepackage{amsfonts,euscript}
\usepackage{amsmath}
\usepackage{amssymb}
\usepackage{amsfonts}
\usepackage{amsthm}
\usepackage{mathrsfs}
\usepackage[all]{xy}
\usepackage{epstopdf}
\usepackage{subfigure}
\usepackage{rotating}
\usepackage[title]{appendix}
\usepackage{appendix}
\usepackage{multirow,rotating}
\usepackage{url}
\usepackage{algorithm}
\usepackage{algorithmicx}
\usepackage{algpseudocode}
\usepackage{subfigure}
\usepackage{setspace}
\usepackage{listings}
\usepackage{color}
\usepackage{tikz}
\usepackage{graphics}
\usepackage[breaklinks=true]{hyperref}
\usepackage[english]{babel}
\usepackage{datetime}
\usepackage{subfigure}
\usepackage{booktabs}
\usepackage{rotfloat}
\usepackage{framed}
\usepackage[authoryear]{natbib}
\usepackage{graphicx}

\usetikzlibrary{chains,shapes,decorations,arrows,calc,arrows.meta,fit,positioning}
\tikzset{
    -Latex,auto,node distance =1 cm and 1 cm,semithick,
    state/.style ={ellipse, draw, minimum width = 0.7 cm},
    point/.style = {circle, draw, inner sep=0.04cm,fill,node contents={}},
    bidirected/.style={Latex-Latex,dashed},
    el/.style = {inner sep=2pt, align=left, sloped}
}

\graphicspath{{pic/}}
\newtheorem{theorem}{Theorem}

\newtheorem{proposition}{Proposition}
\newtheorem{assumption}{Assumption}
\theoremstyle{definition}

\newtheorem{remark}{Remark}

\algdef{SE}[DOWHILE]{Do}{doWhile}{\algorithmicdo}[1]{\algorithmicwhile\ #1}

\allowdisplaybreaks
\begin{document}

\title{Seller-Side Experiments under Interference Induced by Feedback Loops in Two-Sided Platforms}
\author{Zhihua Zhu \\ Tencent \and Zheng Cai \\ Tencent \and Liang Zheng \\ Tencent \and Nian Si\thanks{Corresponding author: niansi@chicagobooth.edu.} \\ Booth School of Business, \\ University of Chicago }
\date{\today }
\maketitle

\begin{abstract}
Two-sided platforms are central to modern commerce and content sharing and often utilize A/B testing for developing new features.  While user-side experiments are common, seller-side experiments become crucial for specific interventions and metrics. This paper investigates the effects of interference caused by feedback loops on seller-side experiments in two-sided platforms, with a particular focus on the counterfactual interleaving design, proposed in \citet{ha2020counterfactual,nandy2021b}. These feedback loops, often generated by pacing algorithms, cause outcomes from earlier sessions to influence subsequent ones.   This paper contributes by creating a mathematical framework to analyze this interference, theoretically estimating its impact, and conducting empirical evaluations of the counterfactual interleaving design in real-world scenarios. Our research shows that feedback loops can result in misleading conclusions about the  treatment effects.
\end{abstract}

\section{Introduction}

Two-sided platforms have increasingly integrated into our daily routines. We utilize shopping platforms like Amazon and Taobao to purchase and sell goods. Video-sharing platforms such as TikTok and Kuaishou allow us to view and upload content. Moreover, platforms like Booking.com and Airbnb have simplified the process of renting out or booking accommodations. A common workflow across these platforms involves users (typically buyers) initiating a request (session), which the platform then processes to match them with a ranked list of sellers or providers.

To ensure optimal user experience, these platforms routinely conduct experiments (A/B tests) before implementing new features. While user-side experiments are predominant, seller-side experiments, also known as supply-side, producer-side, ad-side, creator-side, or video-side experiments in various contexts, are also necessary when user-side experiments are either infeasible or inappropriate.  For instance, certain interventions, such as updating a seller's user interface, can only be applied to sellers. Furthermore, when the objective is to measure metrics like seller retention, seller-side experiments are more suitable.

In naive seller-side experiments, interference often presents more intensely. Interference means that the treatment assignment of some units could affect the outcomes of others \citep{imbens2015causal}.  
 To illustrate this challenge, let consider an advertisement recommendation system in a video-sharing platform. Imagine we're evaluating a new algorithm designed to boost new ads, a new ``cold start'' strategy. In a naive seller-side experiment, let's say we boost 50\% of new ads, which consists of the treatment group, leaving the other half untouched. Due to this boost, ads within the treatment group naturally achieve a higher ranking. Yet, if we were to boost all new ads, the ones in our initial treatment group would actually descend in rank because of the increased volume of videos receiving the same boost. As a result, the data from the experiment could overstate the true impact.

To address this issue, \citet{ha2020counterfactual} and \citet{nandy2021b} propose a counterfactual interleaving design and \citet{Wang-ba-Producer2023} enhance the design with a novel tie-breaking rule to guarantee consistency and monotonicity. In this approach,  a subset of ads is randomly divided into control ads and treatment ads. At the same time, during the ranking phase, both the control strategy and treatment strategy are applied to rank all ads. These two ranking strategies can be referred to as Ranking C and Ranking T, respectively. The results of these two rankings are then merged to produce a final order, M. The merging process uses the order of control ads in Ranking C as their order in M, while the order of treatment ads in Ranking T determines their order in M. If there's a position conflict, it's resolved randomly: one ad retains its spot, while the other is shifted down a slot. Consequently, the placement of ads in the treatment group in the final ranking approximates their rank when all ads are sorted using the treatment strategy. Similarly, the placement of control group ads is nearly identical to their rank under the control strategy.

Following its introduction, this methodology was extensively implemented across major online platforms, such as Facebook, TikTok, and Kuaishou. Although the method appeared promising initially, our implementation and analysis on our platform  revealed substantial interference, particularly in settings with feedback loops. 

In contemporary recommendation systems, rankings in subsequent sessions can be influenced by the rankings of earlier sessions, primarily due to feedback loops. Take advertising recommendations as an example. Each ad campaign has a preset budget, which outlines a specified amount of money to be spent within a given time frame. To adhere to these budget limits, platforms often employs budget control mechanisms \citep{karlsson2020feedback}: they raise the bidding price when the remaining budget is high, and lower it when the budget is nearing depletion. Consequently, if an ad is shown more often in initial sessions, it may be recommended less in later sessions to balance spending, and the opposite is also true. This feedback loop creates the interference.

Feedback loops are not exclusive to the context of advertising. In e-commerce, sales rates must be regulated to align with the inventory level, while on video-sharing platforms, a cold start algorithm for new videos must carefully manage the rate of exposure to balance between discovering new content and promoting popular content. To encompass these various situations, we refer to these feedback control algorithms as "pacing algorithms." A typical workflow for platforms that incorporate a pacing algorithm is depicted in Figure \ref{fig:scheme}.
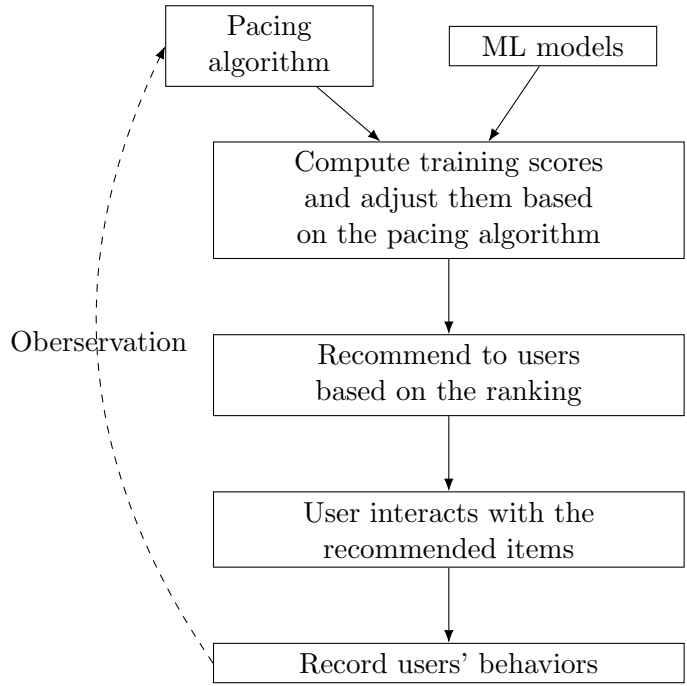
\begin{figure}[ht]
\centering
\begin{tikzpicture}[
		recstate/.style={rectangle, draw, text width=6cm, minimum width=6cm, align=center},recstate2/.style={rectangle, draw, text width=2.5cm, minimum width=2.5cm, align=center}]
		\node[recstate2] (1) {ML models};
  \node[recstate2] (0)  [left = of 1]  {Pacing algorithm};
		\node[recstate] (2) [below= of 1.south west] {Compute training scores and adjust them based on the pacing algorithm};
		\node[recstate] (3) [below= of 2] {Recommend to users based on the ranking};
		\node[recstate] (4) [below =of 3] {User interacts with the recommended items};
		\node[recstate] (5) [below =of 4] {Record users' behaviors};
  \path (0) edge (2);
		\path (1) edge (2);
		\path (2) edge  (3);
		\path (3) edge (4);
		\path (4) edge (5);
		\draw [dashed] (5.west) to[bend left=30] node[midway, above]{Oberservation} (0.west);
	\end{tikzpicture}
\caption{A standard pipeline in platforms with a pacing algorithm.}
\label{fig:scheme}
\end{figure}

Although pacing algorithms are essential in platform operations, they create significant interference that can affect the outcomes of seller-side experiments, including both the naive one and the  counterfactual interleaving designs. In this paper, we will clarify these issues using mathematical models and support our explanations with empirical evidence from the real world. Here, we outline our main contributions:
\begin{enumerate}
\item We develop a framework to represent seller-side experiments that are influenced by interference from feedback loops. We specifically consider the naive design and  the counterfactual  interleaving  design when evaluating different types of features.
\item Using this framework, we illustrate the presence of interference and provide a theoretical estimate of how it biases results.
\item We evaluate the counterfactual interleaving design in a real-world setting affected by feedback loops. Our analysis reveals that feedback loops can lead to incorrect conclusions about the effect of a treatment. This understanding also helps us introduce a straightforward method for detecting such interference in practice.
\end{enumerate}

The remainder of this paper is structured as follows: Section \ref{sec:literature} reviews relevant research on A/B tests under interference, as well as the effects of feedback loops on platform operations. Section \ref{sec:seller-side} provides an overview of the two types of seller-side experiments examined in this paper. In Section \ref{sec:contamination}, we develop mathematical models to investigate the impact of feedback loops on these experiments. Section \ref{sec:numerical} strengthens our analysis with real-world A/B testing data. Finally, we conclude this paper with future work in Section \ref{sec:conclusion}. All the proofs are included in Appendix \ref{append:proof}.

\section{Related Literature}
\label{sec:literature}
\subsection{Interference in Experiments}
The presence of interference is a well-documented phenomenon in academic research. \citet{blake2014marketplace,holtz2020reducing,fradkin2015search} showed empirical evidence that  the bias introduced by interference can be as significant as the treatment effect itself. 

Many researchers have analyzed interference and proposed new experimental designs in two-sided platforms. \citet{johari2022experimental} and \citet{bajari2021multiple}
 introduced two-sided randomizations, also referred to as multiple randomization designs.  By blocking the treatment-control interactions, \citet{ye2023cold} proposed a similar yet different two-sided split design. Additionally, \citet{bright2022reducing} modelled the two-sided marketplaces using a linear programming matching mechanism and developed debiased estimators through shadow prices.
  The concept of bipartite experiments, where treatments are assigned to one group of units and metrics are measured in another, was presented in works by
\citet{eckles2017design,pouget2019variance,harshaw2023design}. Furthermore, the application of cluster experiments in marketplaces was demonstrated in studies by \citet{holtz2020reducing,holtz2020limiting}.

For seller-side experiments, \citet{ha2020counterfactual} and %
\citet{nandy2021b} introduced a counterfactual interleaving design widely
implemented in the industry and \citet{Wang-ba-Producer2023} enhanced the design with a novel tie-breaking rule, as mentioned earlier. 

Regarding advertising experiments, %
\citet{liu2021trustworthy} proposed a budget-split design and \citet{si2022optimal} employed a weighted local linear regression estimation to address imbalances in budget allocation between treatment and control groups.

For interference induced by feedback loops,  \citet{si2023tackling} studied the specific data training loop and proposed a weighted training approach, where \citet{holtz2023study} also studied a similar problem, which they refer to as ``Symbiosis Bias."   
Additionally, \citet{goli2023bias}  developed a bias-correction technique using data from past A/B tests to tackle such interference.  For evaluating bandit learning algorithms, \citet{guo2023evaluating} suggested a two-stage experimental design to assess the treatment effects' lower and upper bounds. Additionally, in the context of search
ranking systems, \citet{musgrave2023measuring}  recommended query-randomized experiments to reduce feature spillover effects.

There are also other types of interference studied in the literature, including Markovian interference  \citep{farias2022markovian,farias2023correcting,glynn2020adaptive,hu2022switchback,li2023experimenting}, temporal interference \citep{bojinov2023design,hu2022switchback,xiong2023data,xiong2023bias,basse2023minimax,xiong2019optimal,ni2023design}, and network interference \citep{hudgens2008toward, gui2015network,li2022random,aronow2017estimating,candogan2023correlated,ugander2013graph,ugander2023randomized,yu2022estimating}.

\subsection{Feedback Loops in Recommendation Systems}
The presence of feedback loops in recommendation systems has been well-documented in the literature for a long time.   \citet{pan2021correcting} discussed the concept of user feedback loops and strategies for mitigating their biases. \citet{jadidinejad2020using} explored the impact of feedback loops on the underlying models. Additionally, \citet{yang2023rectifying} and \citet{khenissi2022modeling} identified fairness issues arising from these loops.
Furthermore, 
 \citet{chaney2018algorithmic}, \citet{mansoury2020feedback}, and \citet{krauth2022breaking} studied how feedback loops amplify homogeneity and popularity biases.  In our work, we specifically focus on the impact of  feedback loops on seller-side experiments.
\section{Seller-Side Experiments}
In the section, we brief talk about the platform pipelines and the implementation of seller-side experiments.

\textbf{Naive seller-side experiments.} In a typical naive seller-side experiment, sellers (such as advertisements, creators, or videos) are randomly assigned to either a treatment group or a control group. The treatment group sellers are equipped with a new feature, while the control group sellers continue using the existing production feature. This new feature could be a user interface modification for the seller or a different campaign budget control algorithm. When a user request is received, a score is calculated for each seller, which may be adjusted by the feedback control algorithms. If the treatment and control groups use different feedback control algorithms, the score adjustments will vary between them. Subsequently, the scores from both groups are combined, and the sellers with the highest scores are recommended to the user. The metrics for users or sellers are then recorded. At the end of the experiment, the treatment effect is calculated by comparing the average metrics of the treatment group with those of the control group.

\textbf{Counterfactual interleaving design.} In the counterfactual interleaving design \citep{ha2020counterfactual,nandy2021b}, sellers are randomly distributed into three groups: treatment, control, and the other group. The inclusion of the other group is essential for minimizing conflicts between treatment and control rankings. This design is distinct from the naive seller-side experiments in terms of how it computes scores and ranks sellers. Both control and treatment strategies are used to calculate seller scores. These scores generate two separate ranking lists, termed Ranking C (control) and Ranking T (treatment). A combined ranking, Ranking M, is then created, positioning control group sellers according to Ranking C and treatment group sellers according to Ranking T. 
\citet{ha2020counterfactual} indicates that conflicts in seller positions are rare when the other group is large, making any tie-breaking rule effective. In contrast, when treatment and control groups are larger, selecting an appropriate tie-breaking rule becomes crucial, as discussed in \citet{nandy2021b,Wang-ba-Producer2023}. After assigning positions to treatment and control sellers, those in the other group fill the remaining slots. This process is illustrated in Figure \ref{fig:counter-interleave}.

\label{sec:seller-side}
\begin{figure}[ht]
    \centering
    \includegraphics[width=9cm]{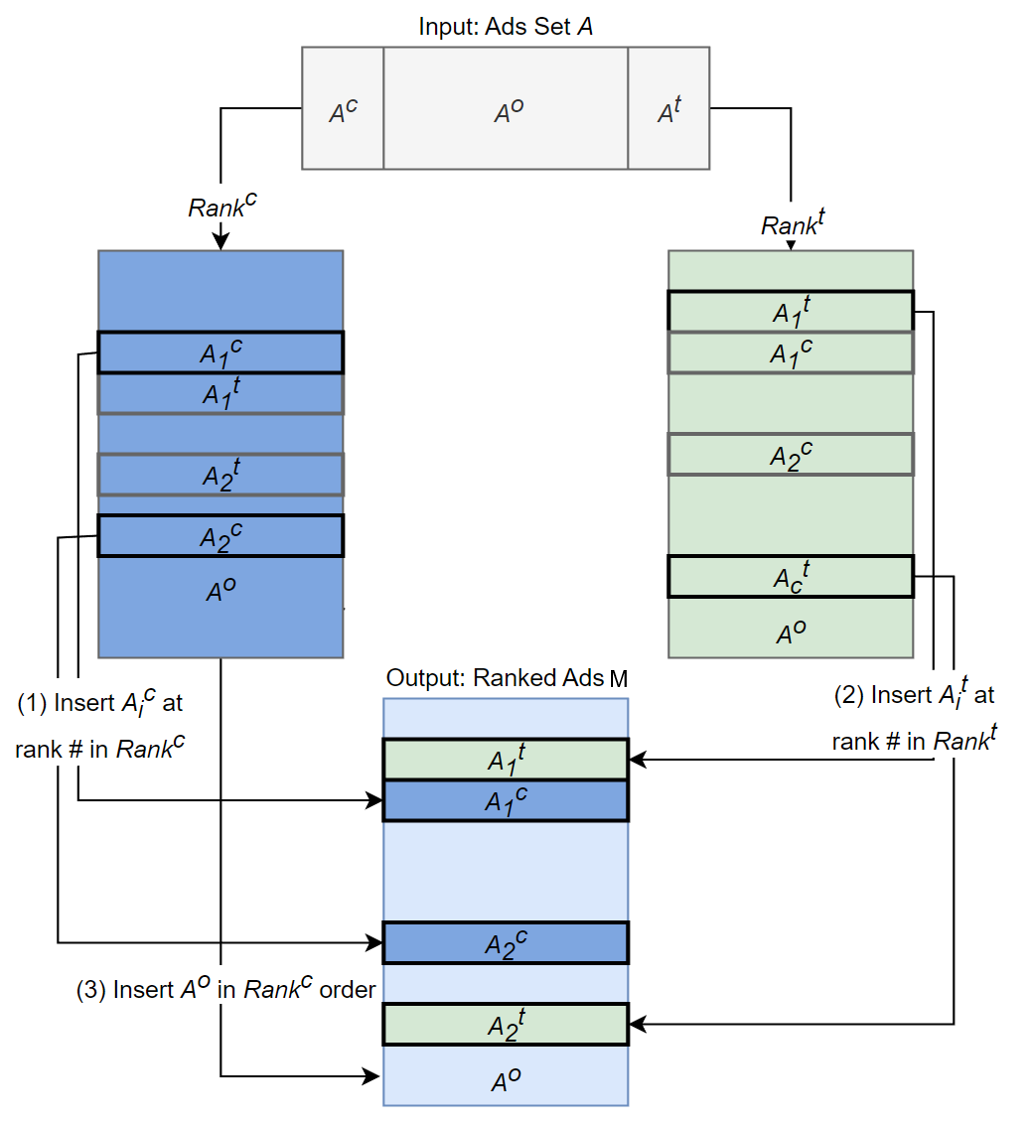}
    \caption{Counterfactual interleave design.}
    \label{fig:counter-interleave}
\end{figure}

\section{A Model for seller-side experiments under Interference Induced by Feedback Loops}
\label{sec:contamination}

When the treatment and control algorithms differ and produce different user-interaction dynamics,  the seller-side experiments would be unreliable and significantly biased by the presence of feedback loops. In this section, we develop a foundational analytical model to analyze seller-side experiments under interference induced by feedback loops.  In Section \ref{subsec:pipeline}, we describe the recommendation system pipeline and the assumptions we make, excluding the  experiment design. Following this, in Section \ref{subsec:without}, we examine seller-side experiments that do not include feedback loops. In contrast, in Section \ref{subsec:with}, we focus on seller-side experiments that incorporate feedback loops.

\subsection{Recommendation System Pipeline}
\label{subsec:pipeline}
We focus on $N$ sellers, denoted by $i=1,2,\ldots ,N.$ The duration of our experiment is $H$. User requests are received continuously over the interval $[0,H]$. The ranking scores of seller $i$ at time $t$, both before and after adjustments made by feedback control algorithms, are represented by $\hat{e}_{i}(t)$ and $r_{i}(t)$, respectively. We use $I_{i}(t)$ as an indicator to signify whether seller $i$ is recommended for a request at time $t$. The true metric for seller $i$ at time $t$, assuming they are recommended, is denoted by $e_{i}(t)$. We define $O_{i}(t)=I_{i}(t)e_{i}(t)$ to represent the observed metrics at time $t$ for seller $i$. 

$S_{i}(t)$ is designated as a function to represent a one-dimensional state process, subject to the influence of feedback loops. In the advertising context, for example, $S_{i}(t)$ could indicate the budget consumption at a specific time $t$. In e-commerce, it might refer to the volume of sales at time $t$. Similarly, in a cold start scenario, $S_{i}(t)$ represents the current number of exposures a video has received. Then, based on the notations defined above, we able to model the dynamics of a recommendation system's pipeline.
\begin{eqnarray}
O_{i}(t) &=&I_{i}(t)e_{i}(t), \\
I_{i}(t) &=&\mathbb{I}\left\{ f_i(\{r_{1}(t),\ldots ,r_{N}(t)\},%
\underline{R}(t))\geq0\right\} , \\
r_{i}(t) &=&\Psi (S_{i}(t),\hat{e}_{i}(t)), \\
\frac{dS_{i}(t)}{dt} &=&\Gamma \left(S_{i}(t), e_{i}(t),r_{i}(t),I_{i}(t),t\right) \label{eq:S},
\end{eqnarray}%
where $\underline{R}(t)$ is some reserve utility for a request at time $t.$
In environments where ads and non-ad content, such as videos, are ranked together, such as on video-sharing platforms, $\underline{R}(t)$ could signify the anticipated utility of recommending a non-advertisement video. 

To avoid the technique issues, the differential equation (\ref{eq:S}) should be understood as the following integral equation:
$$S_{i}(t) =S_{i}(0) + \int_0^t \Gamma \left( S_i(t),e_{i}(s),r_{i}(s),I_{i}(s),s\right)ds.$$
We make the following monotonicity assumptions real-world observations and behaviors. 
\begin{assumption} We assume
    \begin{enumerate}
        \item 
$f_i(\{r_{1},\ldots ,r_{N}\},\underline{R})$   is non-increasing with respect to $%
r_{1},\ldots,r_{i-1},$ $r_{i+1}, \ldots ,r_{N}$   and  $\sum_{i\in\mathcal{I}} f_i(\{r_{1},\ldots ,r_{N}\},\underline{R})$ is non-decreasing with respect to $r_i$, provided that $i\in\mathcal{I}$. 
 \item 
$\Psi (s,\hat{e})$ is non-increasing with respect to $s$.
 \item 
$\Gamma \left( \cdot,\cdot ,\cdot ,\cdot ,\cdot \right) \geq 0$.
\item $\Psi (s,\hat{e})$ is non-decreasing with respect to $\hat{e}$.
    \end{enumerate}
    \label{assump:monotonicity}
\end{assumption}
    Assumption \ref{assump:monotonicity} is clarified as follows: Assumption \ref{assump:monotonicity}.(1) corresponds to the fact that the recommendation system select the highest scores. For example, if one item is recommended, then the system might use the formula 
\begin{equation}f_i(\{r_{1},\ldots ,r_{N}\},\underline{R}) = r_i-\max(r_{1},\ldots,r_{i-1},r_{i+1}, \ldots,r_{N},\underline{R}). \label{eq:formula:max} \end{equation}
    This formula selects the one highest value among the scores.
    Assumptions \ref{assump:monotonicity}.(2) and (3) are based on the fact that $S_i(\cdot)$ represents current consumptions. Consequently, $S_i(\cdot)$ is inherently non-decreasing. Furthermore, pacing algorithms typically exhibit a damping effect; they reduce the ranking of an item when its current consumption is high, and conversely increase it when consumption is low. Finally,  Assumption \ref{assump:monotonicity}.(4)   implies that a higher estimated score prior to adjustment leads to a higher ranking score. 

We plot the causal graph \citep{pearl2000models} in Figure \ref%
{figure:causal} to illustrate the dependence in the feedback loops. 
    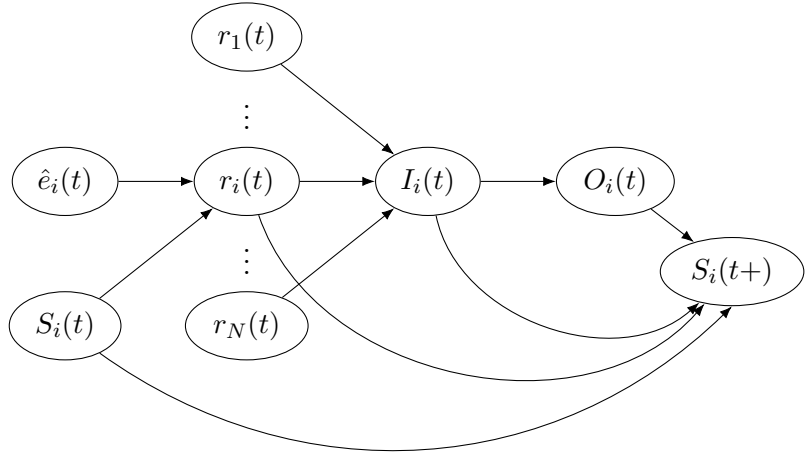
\begin{figure}[ht]
	\centering
	\begin{tikzpicture}
		\node[state] (1) {$\hat{e}_i(t)$};
            
		\node[state] (2) [right =of 1] {$r_i(t)$};
            \node[state] (21) [above =of 2] {$r_{1}(t)$};
            \node[state] (22) [below =of 2] {$r_{N}(t)$};
            
		\node[state] (3) [right =of 2] {$I_i(t)$};
            \node[state] (3') [right =of 3] {$O_i(t)$};
		\node[state] (4) [right =of 3',yshift=-1.2cm,xshift=-1.2cm] {$S_i(t+)$};
		\node[state] (5) [below =of 1,] {$S_i(t)$};

            \node (dot) at ($(2)!.5!(21)$) {$\vdots$};
            \node (dot2) at ($(2)!.5!(22)$) {$\vdots$};
		\path (1) edge   (2);
            \path (21) edge   (3);
            \path (22) edge   (3);
		\path (2) edge   (3);
            \path (3) edge   (3');
		\path (3') edge   (4);
		\path (5) edge   (2);
  \path 	(3)   edge[bend right=60]    (4);
		\path 	(2)   edge[bend right=60]    (4);
		\path (5) edge[bend right=40]   (4.south);
	\end{tikzpicture}

	\caption{Dependence of different objects in the feedback loops.}
	\label{figure:causal}
\end{figure}

Our goal is to estimate the Global Treatment Effect (GTE), defined as the difference in metrics observed under global treatment and global control conditions. We describe this process as follows: within the global treatment condition, the procedure is outlined as
\begin{eqnarray}
O_{i}^{GT}(t) &=&I_{i}^{GT}(t)e_{i}^{GT}(t), \label{GT:O}\\
I_{i}^{GT}(t) &=&\mathbb{I}\left\{f_i(\{r_{1}^{GT}(t),\ldots
,r_{N}^{GT}(t)\},\underline{R}(t))\geq0\right\} , \label{GT:I} \\
r_{i}^{GT}(t) &=&\Psi^T (S_{i}^{GT}(t),\hat{e}_{i}^{GT}(t)), \label{GT:r} \\
\frac{dS_{i}^{GT}(t)}{dt} &=&\Gamma \left(S_{i}^{GT}(t),
e_{i}^{GT}(t),r_{i}^{GT}(t),I_{i}^{GT}(t),t\right). \label{GT:S}
\end{eqnarray}
Similarly,  within the global control regime, we have:
\begin{eqnarray}
O_{i}^{GC}(t) &=&I_{i}^{GC}(t)e_{i}^{GC}(t),  \label{GC:O}\\
I_{i}^{GC}(t) &=&\mathbb{I}\left\{  f_i(\{r_{1}^{GC}(t),\ldots
,r_{N}^{GC}(t)\},\underline{R}(t))\geq0\right\} , \label{GC:I} \\
r_{i}^{GC}(t) &=&\Psi^C (S_{i}^{GC}(t),\hat{e}_{i}^{GC}(t)), \label{GC:r}\\
\frac{dS_{i}^{GC}(t)}{dt} &=&\Gamma \left(S_{i}^{GC}(t),
e_{i}^{GC}(t),r_{i}^{GC}(t),I_{i}^{GC}(t),t\right) \label{GC:S}.
\end{eqnarray}
Here, we assume $S_{i}^{GC}(0)=S_{i}^{GT}(0)$ and $r_{i}^{GC}(0)=r_{i}^{GT}(0)$, for $i=1,2,\ldots,N$.
The GTE is thereby defined as 
\begin{equation}
	\mathrm{GTE}=\frac{1}{N}\sum_{i=1}^{N}\int_{0}^{H}\left( O_{i}^{GT}(t)-O_{i}^{GC}(t)\right)
dt .
\label{GTE}
\end{equation}
\subsection{Seller-Side Experiments without Feedback Loops}
\label{subsec:without}
We consider an experiment that randomly assigns a seller to the treatment group with probability $p$ and to the control group with probability $1-p$. We use $\mathcal{T}$ and $\mathcal{C}$ to denote the treatment and control groups respectively: $\mathcal{T} \cap \mathcal{C} = \emptyset$ and $\mathcal{T} \cup \mathcal{C} = \{1,2,\ldots,N\}$.  We do not specifically consider the other group since  the other group could be absorbed into the reserve utility $\underline{R}(\cdot)$ and we ignore position conflicts in our model. However, incorporating the other group will not change any results we shall present. 

We will consider two types of treatment features: improving item performance $e_i(t)$, for example, 1) enhancing a seller's content and product descriptions to increase customer purchase likelihood; and 2) testing the ranking algorithm $r_i(t)$, such as promoting certain sellers or making the estimated $\hat{e}_i(t)$ more closely align with the true metrics ${e}_i(t)$.

\subsubsection{Naive seller-side experiments}
We first write down the dynamics for  naive seller-side experiment without interference by feedback loops:
\begin{align}
O_{i}^{E}(t) &=I_{i}^{E}(t)e_{i}^{E}(t), \label{NEwithoutO}\\
I_{i}^{E}(t) &=\mathbb{I}\left\{ f_i^{NE}(t)\geq0\right\}, \label{NEwithoutI}
\\
r_{i}^{E}(t) &=\hat{e}_{i}^{E}(t),
\end{align}%
where $E\in\{C,T\}$ indicates the control or treatment assignment, with $E=T$ for members of the treatment group $i\in\mathcal{T}$ and $E=C$ for 
  those in the control group $i\in\mathcal{C}$, and the threshold function is defined
\begin{equation}
f_i^{NE}(t)= f_i\left(\left\{ r_{i}^{T}(t),i\in \mathcal{T}\right\} \cup \left\{
r_{i}^{C}(t),i\in \mathcal{C}\right\} ,\underline{R}(t)\right)
\label{NEwithoutf}
\end{equation}

Th GTE estimator is defined as 
\begin{equation}
\label{GTE_hat}
    \widehat{GTE} = 	\frac{1}{Np}\sum_{i \in \mathcal{T}} \int_{0}^{H} O_{i}^{T}(t)dt-\frac{1}{N(1-p)}\sum_{i \in \mathcal{C}} \int_{0}^{H} O_{i}^{C}(t)
dt .
\end{equation}
In the dynamics (\ref{NEwithoutO})-(\ref{NEwithoutf}), we simplify the process by excluding the state $S_i(\cdot)$ and directly equate the ranking score $r_i(t)$  to the estimated score $\hat{e}_i(t)$  from the ML models. Additionally, the threshold function $f_i^{NE}(\cdot)$ (where NE stands for \textbf{N}aive seller-side \textbf{E}xperiments)  is designed to simply aggregate the ranking scores from both the treatment and control groups.

 \textbf{Case 1:} Testing item performance $e_i(t)$. In this case we assume $\hat{e}_i^C(t)=\hat{e}_i^T(t)$ for $i=1,2,\ldots, N$ and $t\in[0,H]$. Proposition \ref{prop:NEwithoutcase1} below shows that the naive seller-side experiments are unbiased in this case without feedback loops.

 \begin{proposition}
     \label{prop:NEwithoutcase1}
     Suppose that $\hat{e}_i^C(t)=\hat{e}_i^T(t)$ for $i=1,2,\ldots, T$ and $t\in[0,H]$. Then, we have the GTE estimator (\ref{GTE_hat}) is unbiased of the true GTE (\ref{GTE}), i.e.,
     $
     \mathbb{E}[\widehat{GTE}]=\mathrm{GTE}.
     $
 \end{proposition}

 Proposition \ref{prop:NEwithoutcase1} suggests that if experimenters believe the treatment affects only the true item performance without impacting the rankings, and the effects from feedback loops are negligible, then using the naive design is sufficient.

 \textbf{Case 2:} Testing the ranking algorithm $r_i(t)$, where we assume $\hat{e}_i^C(t)\neq\hat{e}_i^T(t)$ in general. In this case, this naive design is biased because $I_i^T(t) \neq I_i^{GT}(t),I_i^C(t) \neq I_i^{GC}(t)$   generally lead to different items being recommended in the experimental versus global treatment/control regimes. For instance, if we consider boosting specific sellers, which is represented by $\hat{e}_i^T(t)>\hat{e}_i^C(t)$. Then, the treatment sellers would typically be ranked higher, exhibiting a cannibalization effects. then sellers under treatment would generally be ranked higher, leading to cannibalization effects. As a result, the estimator tends to overestimate the true GTE, that is,
 $$ \mathbb{E}[\widehat{GTE}]\geq \mathrm{GTE}.$$
 This phenomenon aligns with the "cold start" examples discussed in the Introduction. However, we will see that the counterfactual interleaving design can address this issue, provided that the system is not influenced by feedback loops.

\subsubsection{Counterfactual interleaving design.}
The counterfactual interleaving design differs from the naive design solely in the dynamics of the recommendation item $I$
I, as defined in the equation (\ref{NEwithoutI}). Here, $I$ follows 
\begin{align}
 I_{i}^{T}(t) &=\mathbb{I}\left\{  f_i^{CE,T}(t)\geq0\right\}, \text{ for } i \in \mathcal{T} \text{ and } \label{CEwithoutIT} \\
I_{i}^{C}(t) &=\mathbb{I}\left\{  f_i^{CE,C}(t)\geq0\right\}, \text{ for } i \in \mathcal{C}    \label{CEwithoutIC}
\end{align}
where 
\begin{align}
    f_i^{CE,T}(t)&=f_i\left(\left\{ r_{i}^{T}(t),i=1,2,\ldots,N \right\}  ,\underline{R}(t)\right) \text{ and } \label{CEwithoutfT} \\
    f_i^{CE,C}(t)&=f_i\left(\left\{ r_{i}^{C}(t),i=1,2,\ldots,N \right\}  ,\underline{R}(t)\right).   \label{CEwithoutfC}
\end{align}
Here, CE stands for \textbf{C}ounterfactual interleaving \textbf{E}xperiments. To compute  $f_i^{CE,T}(t)$ and $f_i^{CE,C}(t)$, it is necessary to maintain two versions of ranking scores for all sellers: one for the treatment group and one for the control group.
\begin{equation}
    r_i^T(t)=\hat{e}_i^T(t) \text{ and } r_i^C(t)=\hat{e}_i^C(t),
    \label{CEwithoutr}
\end{equation}
for all $i=1,2,\ldots, N$ and $t\in[0,H]$.
\begin{remark}
     In our continuous model, we can assume there are no position conflicts almost everywhere. Therefore, In dynamics (\ref{CEwithoutIT}) and (\ref{CEwithoutIC}), it is safe to ignore the conflicts.
\end{remark}
Proposition \ref{prop:CEwithoutunbiased} below demonstrates  that the counterfactual interleaving designs are unbiased for both types of treatment features, updating $e_i(t)$ and $\hat{e}_i(t)$. 
\begin{table*}[!ht]
  \centering
  \caption{Comparing the bias in naive seller-side experiments and counterfactual interleaving design  for different types of treatment features in the absense of feedback loops in the system.}

    \begin{tabular}{lcc}
    \toprule
     & Naive seller-side experiments & Counterfactual interleaving design \\
     \midrule
    Testing item performance & \color{green}unbiased & \color{green}unbiased \\
    Testing the ranking algorithm & \color{red}biased & \color{green}unbiased \\
    \bottomrule
    \end{tabular}%
  \label{tab:withfeedbackloop}%
\end{table*}
\begin{proposition}
\label{prop:CEwithoutunbiased}
    Suppose that the experiments follow the dynamics (\ref{NEwithoutO}), (\ref{CEwithoutIT})-(\ref{CEwithoutr}). And we allow for $e_i^T(t)\neq e_i^C(t)$ and $\hat{e}_i^T(t)\neq \hat{e}_i^C(t)$.  Then, we have the GTE estimator (\ref{GTE_hat}) is unbiased of the true GTE (\ref{GTE}), i.e.,
     $
     \mathbb{E}[\widehat{GTE}]=\mathrm{GTE}.
     $
\end{proposition}
Table \ref{tab:withfeedbackloop} summarizes all four cases discussed in this subsection, clearly indicating that interference is not severe in the absence of feedback loops in the system and that counterfactual interleaving designs can provide unbiased estimations.
\subsection{Seller-Side Experiments with Feedback Loops}
\label{subsec:with}
In this subsection, we will incorporate the state process (\ref{eq:S}) in our model and evaluate the biases of naive seller-side experiments and the Counterfactual interleaving design for testing different types of treatment features. To the ease of exposition, we assume 
\begin{equation}
  \Gamma \left( S_{i}(t), e_{i}(t),r_{i}(t),I_{i}(t),t\right) = O_{i}(t)=I_{i}(t)e_{i}(t).
  \label{assump:Gamma}
\end{equation}
Note that the e-commerce and cold start examples fit this assumption well. In the e-commerce scenario, 
 $e_i(t)$ could represent a customer's decision to purchase an item from seller $i$ at time $t$, and $S_i(t)$ would denote the cumulative sales of seller $i$ up to time $t$. In the cold start example,  $e_i(t)$ might model a single view of a new video $i$, with $S_i(t)$ indicating the total exposure of seller $i$ up to time $t$. Furthermore, in the advertising example, the state process dynamics slightly differ from the assumption $(\ref{assump:Gamma})$, as $dS_i(t)/dt$ is usually the payment for that auction, while $e_i(t)$ involves metrics like conversions, installs, or clicks. Nevertheless, as we shall see in Section \ref{sec:numerical}, similar results will still hold. In the following discussion of this section, we will use the sales rate control in e-commerce as a running example.
\subsubsection{Naive seller-side experiments}
We rewrite the dynamics for  naive seller-side experiment with interference by feedback loops:
\begin{align}
O_{i}^{E}(t) &=I_{i}^{E}(t)e_{i}^{E}(t), \label{NEwithO}\\
I_{i}^{E}(t) &=\mathbb{I}\left\{ f_i^{NE}(t)\geq0\right\}, \label{NEwithI}
\\
r_{i}^{E}(t) &=\Psi^E (S_{i}^{E}(t),\hat{e}_{i}^{E}(t)), \\
\frac{dS_{i}^{E}(t)}{dt} &=O_{i}^{E}(t), \label{NEwithS}
\end{align}%
where $f_i^{NE}(t)$ is the same function defined in Equation (\ref{NEwithoutf}).

 \textbf{Case 1:} Testing item performance $e(t)$. As usual, we assume $\hat{e}_i^C(t)=\hat{e}_i^T(t)$ for $i=1,\ldots, N$, $t\in[0,H]$ and $\Psi^T(\cdot,\cdot)=\Psi^C(\cdot,\cdot)$. Furthermore, for simplicity, we assume a uniform treatment effect:
 \begin{assumption}
    ${e}_i^T(t)\geq {e}_i^C(t)$, for $i=1,\ldots, N$, $t\in[0,H]$. 
     \label{assump:uniform:e}
 \end{assumption}
The following technical assumption is used in the proof.
 \begin{assumption}
${e}_i(\cdot),\hat{e}_i(\cdot),f_i(\cdot),\Phi$ are continuous and the events $f_i^{NE}(t)=0$,  $f_i^{GT}(t)=0$, $f_i^{GC}(t)=0$ only occur finite times, for $i=1,\ldots,N$, i.e., the following set contains finite elements
     \[
     \left \{t: f_i^{NE}(t)=0 \text{ or } f_i^{GT}(t)=0 \text{ or } f_i^{GC}(t)=0 \text{ for some } i\right \}.
     \]
     \label{assump:technical}
 \end{assumption}

Then, Theorem \ref{thm:NE:with} characterizes that bias directions in this case.
\begin{theorem}
\label{thm:NE:with}
Suppose Assumptions \ref{assump:monotonicity},\ref{assump:uniform:e}  and \ref{assump:technical}   are enforced and $\hat{e}_i^C(t)=\hat{e}_i^T(t)$ for $i=1,2,\ldots, N$ and $t\in[0,H]$, with $\Psi^T(\cdot,\cdot)=\Psi^C(\cdot,\cdot)$.  We consider a naive seller-side experiment following the dynamics (\ref{NEwithO})-(\ref{NEwithS}). Then, we have 
\begin{align}
  \int_{0}^{H} O_{i}^{C}(t)dt &\geq \int_{0}^{H} O_{i}^{GC}(t)dt \text{ for } i\in\mathcal{C},  \\
   \int_{0}^{H} O_{i}^{T}(t)dt &\leq \int_{0}^{H} O_{i}^{GT}(t)dt \text{ for } i\in\mathcal{T} ,
\end{align}
which means the estimator underestimates the true effects, i.e.,
$$ \mathbb{E}[\widehat{GTE}]\leq \mathrm{GTE}.$$
\end{theorem}
The insights from these results can be intuitively explained as follows. Due to the damping effect inherent in pacing algorithms, a pacing algorithm will slow down if the current consumption is already high. If we consider a treatment feature that enhances performance, resulting in higher consumption for treatment sellers, then the ranking scores for these sellers, after being adjusted by the pacing algorithm, will be lower. Consequently, this adjustment leads to an underestimation of the treatment effects.
\begin{table*}[!b]
  \centering
  \caption{Comparing the bias in naive seller-side experiments and counterfactual interleaving design  for different types of treatment features in the presence of feedback loops in the system. Suppose GTE is positive and a pacing algorithm is used.}

    \begin{tabular}{lcc}
    \toprule
     & Naive seller-side experiments & Counterfactual interleaving design \\
     \midrule
    Testing item performance & underestimate & underestimate \\
    Testing the ranking algorithm & unclear & underestimate \\
    \bottomrule
    \end{tabular}%
  \label{tab:withoutfeedbackloop}%
\end{table*}

 \textbf{Case 2:} Testing the ranking algorithm $r_i(t)$. In this case, let consider  a scenario where 
 $\Psi^T (s,\hat{e}_{i}^{E}(t)) \neq \Psi^C (s,\hat{e}_{i}^{C}(t)).$
This could represent two possible situations:
 \begin{enumerate}
     \item \textbf{Alternating the feedback loop dependence.} For instance, if a platform wishes to experiment with a quicker sales speed for specific categories, it implies that sellers, at any level of cumulative sales, would receive a higher ranking score. Mathematically, this is depicted as
     $$ \Psi^T (s,\hat{e}) \geq \Psi^C (s,\hat{e}) \text{ and } \hat{e}_{i}^{T}(t) = \hat{e}_{i}^{C}(t).$$
     \item \textbf{Modifying  the estimation score.} This could mean the platform aims to boost certain sellers to achieve specific campaign objectives. Mathematically, this is written as
     $$ \Psi^T (s,\hat{e}) = \Psi^C (s,\hat{e}) \text{ and } \hat{e}_{i}^{T}(t) \geq \hat{e}_{i}^{C}(t).$$
     Furthermore, the platform might develop new machine learning models to ensure that the estimated $\hat{e}_i(t)$ aligns more accurately with the actual metrics ${e}_i(t)$. in this case, 
      $$ \Psi^T (s,\hat{e}) = \Psi^C (s,\hat{e}),\hat{e}_{i}^{T}(t) \neq \hat{e}_{i}^{C}(t) \text{ and } \sum_{i=1}^{N} O_{i}^{GT}(t)\geq \sum_{i=1}^{N} O_{i}^{GC}(t).$$
 \end{enumerate}
 In this scenario,  the direction of bias   is unclear. Let consider $$\Psi^T (s,\hat{e}_{i}^{E}(t)) \geq \Psi^C (s,\hat{e}_{i}^{C}(t)).$$ On the one hand, as discussed in Section \ref{subsec:without}, due to higher initial ranking scores before adjustment, sellers under treatment would generally receive higher rankings, leading to cannibalization effects. However, on the other hand, the damping effect present in pacing algorithms results in lower ranking scores, which conversely causes the ranking of treatment sellers to drop.
\subsubsection{Counterfactual interleaving design.}
The dynamics of the counterfactual interleaving design in the presence of feedback loops are described below. 
\begin{align}
O_{i}^{E}(t) &=I_{i}^{E}(t)e_{i}^{E}(t), \label{CEwithO}\\
I_{i}^{T}(t) &=\mathbb{I}\left\{ f_i^{CE,T}(t)\geq0\right\}, \text{ for } i \in \mathcal{T} \text{ and } \label{CEwithIT} \\
I_{i}^{C}(t) &=\mathbb{I}\left\{  f_i^{CE,C}(t)\geq0\right\}, \text{ for } i \in \mathcal{C},    \label{CEwithIC}
\\
r_i^T(t)&=\Psi^E (S_{i}^{E}(t),\hat{e}_{i}^{C}(t)) \text{ and } r_i^C(t)=\Psi^E (S_{i}^{E}(t),\hat{e}_{i}^{T}(t)), \\
\frac{dS_{i}^{E}(t)}{dt} &=O_{i}^{E}(t), \label{CEwithS}
\end{align}%
where $f_i^{CE,T}(t)$ and $f_i^{CE,C}(t)$ are  defined in (\ref{CEwithoutfT}) and  (\ref{CEwithoutfC}). 

\textbf{Case 1:} Testing item performance $e_i(t)$. Same as above, we assume $\hat{e}_i^C(t)=\hat{e}_i^T(t)$ for $i=1,2,\ldots, N$ and $t\in[0,H]$. 
\begin{theorem}
\label{thm:CE:with1}
Suppose Assumptions \ref{assump:monotonicity}, \ref{assump:uniform:e} and \ref{assump:technical}\footnote{We replace $f^{NE}(\cdot)$ with $f^{CE,T}(\cdot)$   and $f^{CE,C}(\cdot)$.}  are enforced and $\hat{e}_i^C(t)=\hat{e}_i^T(t)$ for $i=1,2,\ldots, N$ and $t\in[0,H]$, with $\Psi^T(\cdot,\cdot)=\Psi^C(\cdot,\cdot)$.  We consider a counterfactual interleaving design following the dynamics (\ref{CEwithO})-(\ref{CEwithS}). Then, we have 
\begin{align}
  \int_{0}^{H} O_{i}^{C}(t)dt &\geq \int_{0}^{H} O_{i}^{GC}(t)dt \text{ for } i\in\mathcal{C},  \\
   \int_{0}^{H} O_{i}^{T}(t)dt &\leq \int_{0}^{H} O_{i}^{GT}(t)dt \text{ for } i\in\mathcal{T} ,
\end{align}
which means the estimator underestimates the true effects, i.e.,
$$ \mathbb{E}[\widehat{GTE}]\leq \mathrm{GTE}.$$
\end{theorem}
The proof follows from Theorem \ref{thm:NE:with} by noting that $r_i^T(t)=r_i^C(t)$ and $f_i^{CE,T}(\cdot)=f_i^{CE,C}(\cdot)=f_i^{NE}(\cdot)$.

\textbf{Case 2:} Testing the ranking algorithm $r_i(t)$. For the ease of exposition, we assume the following uniform treatment effects on the ranking function.
 \begin{assumption}
     We assume there exists $\alpha>1$ such that
     \begin{align*}
     &\Psi^T (s,\hat{e}_{i}^{E}(t)) = \alpha \Psi^C (s,\hat{e}_{i}^{C}(t)), \text{ and} \\
     &f_i(\alpha r_1,\ldots, \alpha r_N) \geq f_i(r_1,\ldots,  r_N) \text { if } f_i(r_1,\ldots,  r_N)\geq 0,
     \end{align*}
     for $s\geq0,i=1,2,\ldots, N$, and $t\in[0,H]$. 
     \label{assump:uniform:psi}
 \end{assumption}
Assumption \ref{assump:uniform:psi} guarantees that an item recommended by the control algorithm would also be recommended by the treatment algorithm. Function (\ref{eq:formula:max}) meets this requirement.
The following theorem states that the estimator obtained from counterfactual interleaving design will also underestimate the true GTE when testing the ranking algorithm  in the presence of feedback loops. 
\begin{theorem}
\label{thm:CE:with2}
Suppose Assumptions \ref{assump:monotonicity}, \ref{assump:technical}\footnote{We replace $f^{NE}(\cdot)$ with $f^{CE,T}(\cdot)$   and $f^{CE,C}(\cdot)$.}  and \ref{assump:uniform:psi} are enforced and ${e}_i^C(t)={e}_i^T(t)$ for $i=1,2,\ldots, N$ and $t\in[0,H]$.  We consider a counterfactual interleaving design following the dynamics (\ref{CEwithO})-(\ref{CEwithS}). Then, we have 
\begin{align}
  \int_{0}^{H} O_{i}^{C}(t)dt &\geq \int_{0}^{H} O_{i}^{GC}(t)dt \text{ for } i\in\mathcal{C}  \\
   \int_{0}^{H} O_{i}^{T}(t)dt &\leq \int_{0}^{H} O_{i}^{GT}(t)dt \text{ for } i\in\mathcal{T} ,
\end{align}
which means the estimator underestimates the true effects, i.e.,
$$ \mathbb{E}[\widehat{GTE}]\leq \mathrm{GTE}.$$
\end{theorem}

Let's explain these results in an intuitive manner. Consider a scenario where the treatment algorithms employ a more rapid pacing speed. This implies that at any given time, the consumption of inventories under the treatment algorithms would typically exceed that of the control algorithms. In the counterfactual interleaving design, this would mean that treatment sellers consume more than their control counterparts on average. When the treatment algorithm is employed across all items to derive Ranking T, treatment items, on average, tend to rank lower compared to their positions in the global treatment setting. This is primarily due to the damping effect inherent in pacing algorithms. As a result, the treatment ranking realized in the experiments may not accurately reflect the ranking in the global treatment regime.


Table \ref{tab:withoutfeedbackloop} summarizes all four cases discussed in this subsection.   It is evident that interference is always present when the system is affected by feedback loops. Additionally, the counterfactual interleaving design typically underestimates the true effect if a pacing algorithm with damping effects is employed.

\textbf{Detecting interference induced by feedback loops.} To this end, we propose a simple method to detect the interference induced by feedback loops in the counterfactual interleaving design. We will compare the average rankings of treatment sellers and control sellers when both are under treatment algorithms,  i.e., comparing $\{r_i^T,i \in \mathcal{T}\}$ with  $\{r_i^T,i \in \mathcal{C}\}$. If these averages are similar, then the interference is likely not significant. However, if there is a noticeable deviation as the experiment progresses, the interference may render the experiment invalid. Similarly, we can also compare the average rankings of the treatment sellers and control sellers both under the control algorithms, i.e., comparing $\{r_i^C,i \in \mathcal{T}\}$ with  $\{r_i^C,i \in \mathcal{C}\}$. We shall implement this method in our empirical study.

\section{Empirical Results based on Real-World A/B tests}
\label{sec:numerical}

We partnered with the advertising recommendation team at Tencent, a world-class content-sharing platform. In the advertising recommendations, it is commonly observed the estimated scores overestimate the real effect, due to perhaps maximization bias, especially for those ads with low impressions \citep{fan2022calibration}. To address this bias, Tencent implemented a strategy: at any time 
$t$, if the cumulative realized value (clicks, conversions, etc.) up to that moment surpasses the cumulative estimated value, the current estimated scores are adjusted upward. Conversely, if the overall realized value falls short of the cumulative estimated value, the scores are decremented.

This type of feedback loop differs slightly from the pacing algorithm model discussed in this paper. However, the results and insights derived should be similar. In mathematical terms, 
$e_i(t)$ is the raw estimated score for the $i$-th ad in relation to a request at time $t$. Additionally,  $s_i(\cdot)$  is the state process that captures the cumulative overestimation (or underestimation)  for the $i$-th ad up until time $t$, where
$$s_i(t) = \left. \left(\int_0^t r_i(s)ds\right)\right/\left(\int_0^t e_i(s)ds\right).$$
Given these, the ranking score can be defined as 
$r_i(t) = \lambda_i(s_i(t))e_i(t)$, for  $i = 1,\ldots, N,$ 
where $\lambda_i(s_i(t))$ represents the adjustment factor. At the beginning of a day, $\lambda$ is initialized at  1. Further, $\lambda_i(s)$ decreases as $s$ increases, and $\lambda_i(s_i(t))<1$  adjusts for overestimation  ($s_i(t)>1$), while $\lambda_i(s_i(t))>1$  adjusts for  underestimation ($s_i(t)<1$).

Our empirical observation reveals a challenge of this adjusting mechanism:  the values of 
 $\lambda$s tend to fluctuate significantly due to the inherent randomness in realized values. Such volatility can negatively impact the overall performance. To address this, we devised a new strategy to constrain the variability of $\lambda$, effectively reducing its swing or ``effective range.'' 
 
In our experiments, we'll contrast this new strategy with the original approach using the counterfactual interleaving design. More precisely, we'll be comparing the treatment adjustment method $\lambda^T(\cdot)$ with the control adjustment method 
 $\lambda^C(\cdot)$, where we allocate 10\% control ads and 10\% treatment ads.  Specifically, the platform ranks and charges based on the scores
\begin{align*}
r^T_i(t) = \lambda^T_i(s_i(t))e_i(t), \text{ and } r^C_i(t) = \lambda^C_i(s_i(t))e_i(t),
\end{align*}
\text{ for } i$ = 1,2,\ldots, N$, and $|\lambda_i^T(s)-1|<|\lambda_i^C(s)-1|$ for any $s \in \mathbb{R}_+$.

Despite simulations and A/A tests consistently indicating the superiority of the treatment strategy over the control strategy, the counterfactual interleaving design suggests otherwise: Table \ref{tab:addlabel} shows that the estimation from the counterfactual interleaving experiments suggest a  strongly negative effects.

We attribute this discrepancy to interference arising from feedback loops. Given the systematic overestimation and our specific treatment strategy, it's anticipated that  $ \lambda_i^T(s_i(t-))$ would typically exceed $\lambda_i^C(s_i(t-))$.  Consequently, $s^T_i(t-) > s^C_i(t-)$ in general. Due to the inverse monotonic behavior of 
 $\lambda (\cdot)$, this means that $ \lambda_i^T(s^T_i(t-))<\lambda_i^T(s^C_i(t-))$ and $ \lambda_i^C(s^T_i(t-))<\lambda_i^C(s^C_i(t-))$, causing that the treatment ads frequently rank lower, while control ads often rank higher in the experiment. To bolster our rationale, we visualize the average $\lambda^T$  values under the treatment strategy, $\bar{\lambda}^T$s, for the control ads, treatment ads, and other ads over time in Figure  \ref{fig:average_lambda}. The figure elucidates that while the $\bar{\lambda}^T$s are nearly identical across the three groups initially, the $\bar{\lambda}^T$s  in the treatment group noticeably diminishes towards the day's end.

\begin{table*}[!h]
	\centering
	\caption{The experimental results using the counterfactual interleaving design}
	
	\begin{tabular}{cccccc}
		\toprule
		\multicolumn{2}{c}{Advertising cost (consumption)} & \multicolumn{2}{c}{Views} & \multicolumn{2}{c}{Gross merchandise value (GMV)} \\
		
		Estimator & Confidence Interval & Estimator & Confidence Interval & Estimator & Confidence Interval \\
		\midrule
		-23\% & [-34\%,-12\%] & -27\% & [-38\%,-15\%] & -21\% & [-34\%,-9\%] \\
		\bottomrule
	\end{tabular}%
	\label{tab:addlabel}%
\end{table*}%

\begin{figure*}[!h]
	\centering
	\includegraphics[width=0.75\textwidth]{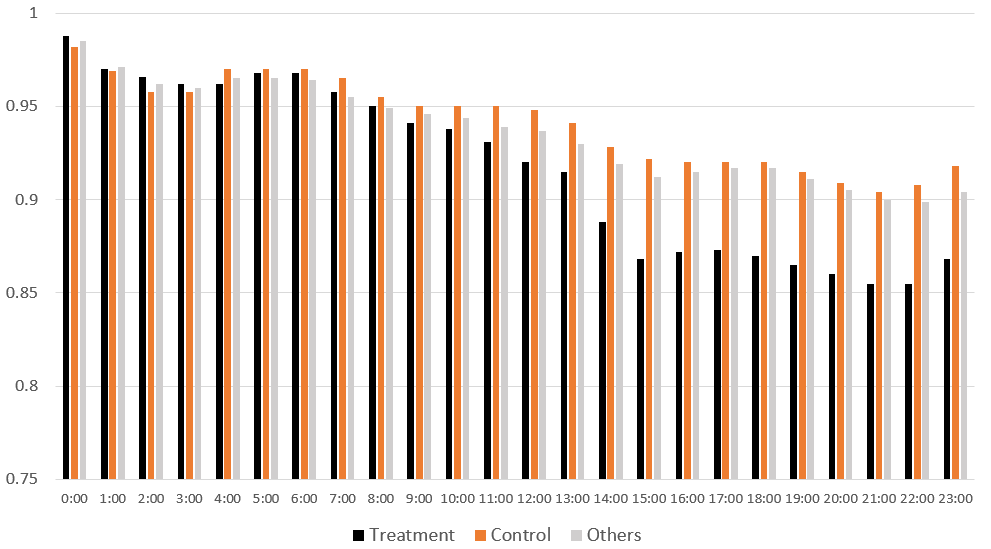}
	\caption{Average $\lambda^T$ values, $\bar{\lambda}^T$s, for the control ads, treatment ads, and other ads over time in a day.}
	\label{fig:average_lambda}
\end{figure*}
\section{Conclusion}
In this paper, we develop a framework to study seller-side experiments affected by interference from feedback loops. Our findings indicate that counterfactual interleaving designs tend to underestimate the true treatment effects when a damping pace algorithm is employed. Additionally, we introduce a method for detecting interference by comparing the rankings of treatment and control groups under the same treatment/control algorithms. As a direction for future research, it would be valuable to explore how experiments can be conducted optimally in scenarios with feedback loops.
\label{sec:conclusion}
\bibliographystyle{ACM-Reference-Format}
\bibliography{feedback,mybib}
\appendix
\newpage
\section{Proofs}
\label{append:proof}
 \begin{proof} [Proof of Proposition \ref{prop:NEwithoutcase1}]
     Recall in this case without feedback loops, we have $\Phi^T(s,\hat{e})=\Phi^C(s,\hat{e})=\hat{e}$. Therefore, we have
     $$
     r_i^{GT}(t)=r_i^{T}(t)=r_i^{C}(t)=r_i^{GC}(t),
     $$
     for $i=1,2,\ldots, N$ and $t\in[0,H]$, because of 
     $$\hat{e}_i^{GT}(t)=\hat{e}_i^T(t)=\hat{e}_i^C(t)=\hat{e}_i^{GC}(t).$$
     It further yields
     $$
     I_i^{GT}(t)=I_i^{T}(t)=I_i^{C}(t)=I_i^{GC}(t),
     $$
     which means
     \begin{align*}
         &\mathbb{E}\left[ \left(\frac{1}{p}\int_{0}^{H} O_{i}^{T}(t)dt\right)\mathbb{I}\{i \in \mathcal{T}\} \right]= \int_{0}^{H} O_{i}^{GT}(t)dt \text{ and} \\
         &\mathbb{E}\left[ \left(\frac{1}{1-p}\int_{0}^{H} O_{i}^{T}(t)dt\right)\mathbb{I}\{i \in \mathcal{C}\} \right]= \int_{0}^{H} O_{i}^{GC}(t)dt,
     \end{align*}
     for $i=1,2,\ldots, N$. 
The desired results then follow.
 \end{proof}
\begin{proof}[Proof of Proposition \ref{prop:CEwithoutunbiased}]
   From (\ref{CEwithoutr}), we know 
   $$
     r_i^{GT}(t)=r_i^{T}(t)\text{ and }r_i^{C}(t)=r_i^{GC}(t),
     $$
     for $i=1,2,\ldots, N$ and $t\in[0,H]$, given that 
     $$\hat{e}_i^{GT}(t)=\hat{e}_i^T(t)\text{ and } \hat{e}_i^C(t)=\hat{e}_i^{GC}(t).$$
     Then, by comparing (\ref{CEwithoutIT}), (\ref{CEwithoutfT}) with (\ref{GT:I}) and  (\ref{CEwithoutIC}), (\ref{CEwithoutfC}) with (\ref{GC:I}), we observe
     $$
     I_i^{GT}(t)=I_i^{T}(t) \text{ and } I_i^{C}(t)=I_i^{GC}(t).
     $$
     The rest of proof follows the proof of Proposition \ref{prop:NEwithoutcase1}.
\end{proof}
\begin{proof}[Proof of Theorem \ref{thm:NE:with}]
    We prove the following claim    
\[
\int_{0}^{t}O_{i}^{E}(s)ds\leq \int_{0}^{t}O_{i}^{GT}(s)ds,\text{ }
\]%
for $i=1,2,\ldots ,N$ and $t\in \lbrack 0,H]$ by contradiction. Suppose that
there exists some $i=$ $1,2,\ldots ,N$ and $t\in \lbrack 0,H]$ such that 
\[
\int_{0}^{t}O_{i}^{E}(s)ds>\int_{0}^{t}O_{i}^{GT}(s)ds.
\]%
We first identify when this inequality occurs for the first time:  
\[
\tau =\inf \left\{ t:\text{ there exists }i\text{ such that }%
\int_{0}^{t}O_{i}^{E}(s)ds>\int_{0}^{t}O_{i}^{GT}(s)ds\right\} ,
\]
and let $\mathcal{I}^{\ast }$ contain the corresponding sellers. That is, for $i^*\in\mathcal{I}^*$, we have
\begin{equation}
\label{ineq:tau_i_star}
\tau =\inf \left\{ t: 
\int_{0}^{t}O_{i^*}^{E}(s)ds>\int_{0}^{t}O_{i^*}^{GT}(s)ds\right\} ,  
\end{equation}%
and for $j\notin\mathcal{I}^*$, there exists $t_0>\tau$ such that for $t\in [\tau,t_0]$,
$$\int_{0}^{t}O_{j}^{E}(s)ds\leq\int_{0}^{t}O_{j}^{GT}(s)ds.$$
By continuity, we have 
\begin{equation}
\int_{0}^{\tau }O_{i^{\ast }}^{E}(s)ds=\int_{0}^{\tau }O_{i^{\ast
}}^{GT}(s)ds,\text{ for }i^{\ast }\in \mathcal{I}^{\ast },
\label{eq:tau:istar}
\end{equation}%
which implies 
$
S_{i^*}^{E}(\tau) =S_{i^*}^{GT}(\tau),
$
\text{ for }$i^{\ast }\in \mathcal{I}^{\ast }.$

Due to the inverse monotonocity in Assumption \ref{assump:monotonicity}.(2),
we have 
\begin{eqnarray}
r_{i^{\ast }}^{E}(\tau ) &= &r_{i^{\ast }}^{GT}(\tau)\text{ for }i^{\ast
}\in \mathcal{I}^{\ast },\text{ }  \label{eq:r:proof:i} \\
r_{j}^{E}(\tau ) &\geq &r_{j}^{GT}(\tau)\text{ for }j\notin \mathcal{I}^{\ast
}.  \label{eq:r:proof:j}
\end{eqnarray}%
Then, by Assumption \ref{assump:monotonicity}.(1), we have%
\begin{eqnarray*}
&&f_{i^{\ast }}^{NE}(\tau ) \\
&=&f_{i^{\ast }}\left( \left\{ r_{j}^{E}(\tau ),j=1,2,\ldots ,N\right\} ,%
\underline{R}(\tau )\right)  \\
&\leq &f_{i^{\ast }}\left( \left\{ r_{j}^{GT}(\tau ),j=1,2,\ldots ,N\right\}
,\underline{R}(\tau )\right)  \\
&=&f_{i^{\ast }}^{GT}(\tau ).
\end{eqnarray*}

We then consider three different cases.

\textbf{Case 1: }There exists some $i^{\ast }\in \mathcal{I}^{\ast }$ such
that 
\[
f_{i^{\ast }}^{GT}(\tau )>0.
\]
By Assumption \ref{assump:technical}, we have there exists $t_{1}>\tau ,$
such that 
\[
f_{i^{\ast }}^{GT}(t)\geq 0,\text{ for }t\in \lbrack \tau ,t_{1}]\ ,
\]%
which means $I_{i^{\ast }}^{GT}(t)=1,$ for $t\in \lbrack \tau ,t_{1}].$
Then, we have for $t \in [\tau,t_1]$%
\begin{eqnarray*}
\int_{0}^{t}O_{i^{\ast }}^{GT}(s)ds &=&S_{i^{\ast }}^{GT}(\tau
)+\int_{\tau }^{t}e_{i^{\ast }}^{GT}(s)ds \\
&\geq &S_{i^{\ast }}^{E}(\tau )+\int_{\tau }^{t}I_{i^{\ast
}}^{E}(s)e_{i^{\ast }}^{E}(s)ds \\
&=&\int_{0}^{t}O_{i^{\ast }}^{E}(s)ds,
\end{eqnarray*}%
which contradicts to the definition of $\tau$, as shown in (\ref{ineq:tau_i_star}).

\textbf{Case 2: }There exists some $i^{\ast }\in \mathcal{I}^{\ast }$ such
that 
\[
0>f_{i^{\ast }}^{NE}(\tau ).
\]%
By Assumption \ref{assump:technical}, we have there exists $t_{2}>\tau ,$
such that 
\[
f_{i^{\ast }}^{NE}(t )<0,\text{ for }t\in \lbrack \tau ,t_{2}]\ ,
\]%
which means $I_{i^{\ast }}^{E}(t)=0,$ for $t\in \lbrack \tau ,t_{2}].$ Then,
we have for $t\in[t,t_2]$%
\begin{eqnarray*}
\int_{0}^{t}O_{i^{\ast }}^{E}(s)ds &=&S_{i^{\ast }}^{E}(\tau ) \\
&\leq &S_{i^{\ast }}^{GT}(\tau )+\int_{\tau }^{t}I_{i^{\ast
}}^{GT}(t)e_{i^{\ast }}^{GT}(s)ds \\
&=&\int_{0}^{t}O_{i^{\ast }}^{GT}(s)ds,
\end{eqnarray*}%
which contradicts to the definition of $\tau$, as shown in (\ref{ineq:tau_i_star}).

\textbf{Case 3: }For all $i^{\ast }\in \mathcal{I}^{\ast }$, we have%
\[
f_{i^{\ast }}^{GT}(\tau )=f_{i^{\ast }}^{NE}(\tau )=0.
\]
If $f_{i^{\ast }}^{GT}(t+)>0$, it reduces to the case 1. Similarly, if $f_{i^{\ast }}^{NE}(t+)<0$, it reduces to the case 2.
Then, there
must exists some $t_{3}\in (\tau ,t_{0})$ such that 
\begin{equation}
f_{i^{\ast }}^{GT}(t)<0\text{ and }f_{i^{\ast }}^{NE}(t)>0,
\label{eq:case3:f}
\end{equation}%
for $t\in (\tau ,t_{3}]$ and $i^{\ast }\in \mathcal{I}^{\ast }.$ 
It further yields  
\begin{eqnarray}
S_{i^{\ast }}^{E}(t) &\geq &S_{i^{\ast }}^{GT}(t),\text{ for }i^{\ast }\in 
\mathcal{I}^{\ast },\text{ }t\in [\tau ,t_{3}]\text{ }  \label{eq:prove:S} \\
S_{j}^{E}(t) &\leq &S_{j}^{GT}(t),\text{ for }j\notin \mathcal{I}^{\ast },%
\text{ }t\in [\tau ,t_{3}].
\end{eqnarray}%
Due to the inverse monotonocity in Assumption \ref{assump:monotonicity}.(2),
we have 
\begin{eqnarray}
r_{i^{\ast }}^{E}(t ) &\leq &r_{i^{\ast }}^{GT}(t)\text{ for }i^{\ast
}\in \mathcal{I}^{\ast },\text{ }t\in [\tau ,t_{3}]  \label{eq:r:proof:i} \\
r_{j}^{E}(t ) &\geq &r_{j}^{GT}(t)\text{ for }j\notin \mathcal{I}^{\ast
},t\in [\tau ,t_{3}].  \label{eq:r:proof:j}
\end{eqnarray}%

Then, by
Assumption \ref{assump:monotonicity}.(1), we have for $t\in (\tau ,t_{0}],$ 
\begin{eqnarray*}
&&\sum_{i^{\ast }\in \mathcal{I}^{\ast }}f_{i^{\ast }}^{NE}(t) \\
&=&\sum_{i^{\ast }\in \mathcal{I}^{\ast }}f_{i^{\ast }}\left( \left\{
r_{j}^{E}(t),j=1,2,\ldots ,N\right\} ,\underline{R}(t)\right)  \\
&\leq &\sum_{i^{\ast }\in \mathcal{I}^{\ast }}f_{i^{\ast }}\left( \left\{
r_{i^{\ast }}^{E}(t),i^{\ast }\in \mathcal{I}^{\ast }\right\} \cup \left\{
r_{j}^{GT}(t),j\notin \mathcal{I}^{\ast }\right\} ,\underline{R}(t)\right) 
\\
&\leq &\sum_{i^{\ast }\in \mathcal{I}^{\ast }}f_{i^{\ast }}\left( \left\{
r_{i^{\ast }}^{GT}(t),i^{\ast }\in \mathcal{I}^{\ast }\right\} \cup \left\{
r_{j}^{GT}(t),j\notin \mathcal{I}^{\ast }\right\} ,\underline{R}(t)\right) 
\\
&=&\sum_{i^{\ast }\in \mathcal{I}^{\ast }}f_{i^{\ast }}^{GT}(t),
\end{eqnarray*}%
which contradicts to the inequality (\ref{eq:case3:f}).

Similarly, we can show 
\[
\int_{0}^{t}O_{i}^{E}(s)ds\geq \int_{0}^{t}O_{i}^{GC}(s)ds,\text{ } 
\]%
for $i=1,2,\ldots ,N$ and $t\in \lbrack 0,H].$

\end{proof}

\begin{proof}[Proof of Theorem \ref{thm:CE:with2}]
    We proceed in the same manner as the proof of Theorem  \ref{thm:NE:with}.

We prove the following claim    
\[
\int_{0}^{t}O_{i}^{E}(s)ds\leq \int_{0}^{t}O_{i}^{GT}(s)ds,\text{ }
\]%
for $i=1,2,\ldots ,N$ and $t\in \lbrack 0,H]$ by contradiction. Suppose that
there exists some $i=$ $1,2,\ldots ,N$ and $t\in \lbrack 0,H]$ such that 
\[
\int_{0}^{t}O_{i}^{E}(s)ds>\int_{0}^{t}O_{i}^{GT}(s)ds.
\]%
We first identify when this inequality occurs for the first time:  
\[
\tau =\inf \left\{ t:\text{ there exists some }i\text{ s.t. }%
\int_{0}^{t}O_{i}^{E}(s)ds>\int_{0}^{t}O_{i}^{GT}(s)ds\right\} ,
\]%
and let $\mathcal{I}^{\ast }$ contain the corresponding sellers. By
continuity, we have 
\begin{equation}
\int_{0}^{\tau }O_{i^{\ast }}^{E}(s)ds=\int_{0}^{\tau }O_{i^{\ast
}}^{GT}(s)ds,\text{ for }i^{\ast }\in \mathcal{I}^{\ast }.
\label{eq:tau:istar2}
\end{equation}%
Similar to Proof of Theorem \ref{thm:NE:with}, we have%
\begin{eqnarray*}
&&f_{i^{\ast }}^{CE,T}(\tau ) \\
&=&f_{i^{\ast }}\left( \left\{ r_{j}^{T}(\tau ),j=1,2,\ldots ,N\right\} ,%
\underline{R}(\tau )\right)  \\
&\leq &f_{i^{\ast }}\left( \left\{ r_{j}^{GT}(\tau ),j=1,2,\ldots ,N\right\}
,\underline{R}(\tau )\right)  \\
&=&f_{i^{\ast }}^{GT}(\tau ).
\end{eqnarray*}

Cases 1 and 2 are essentially the same as in the proof of Theorem \ref{thm:NE:with}.

\textbf{Case 3: }For all $i^{\ast }\in \mathcal{I}^{\ast }$, we have%
\begin{eqnarray*}
f_{i^{\ast }}^{GT}(\tau ) &=&0, \\
f_{i^{\ast }}^{CE,T}(t) &=&0\text{ if }i^{\ast }\in \mathcal{T}\text{ and }%
f_{i^{\ast }}^{CE,C}(t)=0\text{ if }i^{\ast }\in \mathcal{C}\text{.}
\end{eqnarray*}

There must
exists some $t_{3}\in (\tau ,t_{0})$ such that 
\begin{eqnarray}
f_{i^{\ast }}^{GT}(t) &<&0\text{,}  \label{eq:case3:f:GT} \\
f_{i^{\ast }}^{CE,T}(t) &>&0\text{ if }i^{\ast }\in \mathcal{T}\text{ and }%
f_{i^{\ast }}^{CE,C}(t)>0\text{ if }i^{\ast }\in \mathcal{C}\text{, }
\label{eq:case3:f:CE}
\end{eqnarray}%
for $t\in (\tau ,t_{3}]$ and $i^{\ast }\in \mathcal{I}^{\ast }.$

It further yields  
\begin{eqnarray}
S_{i^{\ast }}^{E}(t) &\geq &S_{i^{\ast }}^{GT}(t),\text{ for }i^{\ast }\in 
\mathcal{I}^{\ast },\text{ }t\in [\tau ,t_{3}]\text{ }  \label{eq:prove:S} \\
S_{j}^{E}(t) &\leq &S_{j}^{GT}(t),\text{ for }j\notin \mathcal{I}^{\ast },%
\text{ }t\in [\tau ,t_{3}].
\end{eqnarray}%
Due to the inverse monotonocity in Assumption \ref{assump:monotonicity}.(2),
we have 
\begin{eqnarray}
r_{i^{\ast }}^{T}(t ) &\leq &r_{i^{\ast }}^{GT}(t)\text{ for }i^{\ast
}\in \mathcal{I}^{\ast },\text{ }t\in [\tau ,t_{3}]  \label{eq:r:proof:i2} \\
r_{j}^{T}(t ) &\geq &r_{j}^{GT}(t)\text{ for }j\notin \mathcal{I}^{\ast
},t\in [\tau ,t_{3}].  \label{eq:r:proof:j2}
\end{eqnarray}%

First, by Assumption \ref{assump:uniform:psi}, we have  
\begin{eqnarray*}
&&f_{i^{\ast }}^{CE,C}(t) \\
&=&f_{i^{\ast }}\left( \left\{ r_{j}^{C}(t),j=1,2,\ldots ,N\right\} ,%
\underline{R}(t)\right)  \\
&\leq &f_{i^{\ast }}\left( \left\{ \alpha r_{j}^{C}(t),j=1,2,\ldots
,N\right\} ,\underline{R}(t)\right)  \\
&=&f_{i^{\ast }}\left( \left\{ r_{j}^{T}(t),j=1,2,\ldots ,N\right\} ,%
\underline{R}(t)\right)  \\
&\leq &f_{i^{\ast }}^{CE,T}(t).
\end{eqnarray*}%
Then, inequalities (\ref{eq:case3:f:GT}) and (\ref{eq:case3:f:CE}) imply 
\begin{equation}
f_{i^{\ast }}^{GT}(t)<0\text{, }f_{i^{\ast }}^{CE,T}(t)>0\text{, }
\label{eq:case3:f:both}
\end{equation}%
for $t\in (\tau ,t_{3}]$ and $i^{\ast }\in \mathcal{I}^{\ast }.$

Then, by Assumption \ref{assump:monotonicity}.(1) and the inequalities (\ref%
{eq:r:proof:i2}) and (\ref{eq:r:proof:j2}), we have for $t\in (\tau ,t_{0}],$
\begin{eqnarray*}
&&\sum_{i^{\ast }\in \mathcal{I}^{\ast }}f_{i^{\ast }}^{CE,T}(t) \\
&=&\sum_{i^{\ast }\in \mathcal{I}^{\ast }}f_{i^{\ast }}\left( \left\{
r_{j}^{T}(t),j=1,2,\ldots ,N\right\} ,\underline{R}(t)\right)  \\
&\leq &\sum_{i^{\ast }\in \mathcal{I}^{\ast }}f_{i^{\ast }}\left( \left\{
r_{i^{\ast }}^{T}(t),i^{\ast }\in \mathcal{I}^{\ast }\right\} \cup \left\{
r_{j}^{GT}(t),j\notin \mathcal{I}^{\ast }\right\} ,\underline{R}(t)\right) 
\\
&\leq &\sum_{i^{\ast }\in \mathcal{I}^{\ast }}f_{i^{\ast }}\left( \left\{
r_{i^{\ast }}^{GT}(t),i^{\ast }\in \mathcal{I}^{\ast }\right\} \cup \left\{
r_{j}^{GT}(t),j\notin \mathcal{I}^{\ast }\right\} ,\underline{R}(t)\right) 
\\
&=&\sum_{i^{\ast }\in \mathcal{I}^{\ast }}f_{i^{\ast }}^{GT}(t),
\end{eqnarray*}%
which contradicts to the inequalities (\ref{eq:case3:f:both}).
\end{proof}

\end{document}